\documentclass[conference,a4paper]{IEEEtran}
\usepackage{balance}
\usepackage[utf8]{inputenc}
\usepackage{graphicx}
\usepackage[noadjust]{cite}
\usepackage{amsmath,amsthm,amssymb,mathrsfs}
\usepackage{epstopdf}
\usepackage{threeparttable}
\usepackage{multirow}
\usepackage[colorlinks, citecolor=blue]{hyperref}
\usepackage{color}
\usepackage{caption}
\usepackage{mathtools}
\usepackage{lscape}
\newcounter{MYtempeqncnt}
\usepackage{float}
\usepackage{subfig}
\usepackage{stfloats}
\usepackage{booktabs} 
\newtheorem{theorem}{Theorem}

\newtheorem{proposition}[theorem]{Proposition}

\newtheorem{remark}{Remark}

\usepackage{suffix}
\usepackage{mathtools}

\DeclarePairedDelimiterX\MeijerM[3]{\lparen}{\rparen}%
{\,#3\delimsize\vert\begin{smallmatrix}#1 \\ #2\end{smallmatrix}}
\newcommand\FoxH[8][]{%
  H^{\,#2,#3}_{#4,#5}\FoxH[#1]{#6}{#7}{#8}}
\ifCLASSINFOpdf
\else
\fi

\begin{document}
\title{Effective Rate of RIS-aided Networks with Location and Phase Estimation Uncertainty}
\author{\IEEEauthorblockN{Long Kong{$^\ddagger$},~Steven Kisseleff{$^\dagger$},~Symeon Chatzinotas{$^\dagger$},~Björn Ottersten{$^\dagger$}, and Melike Erol-Kantarci{$^\ddagger$}}\\
\IEEEauthorblockA{$^\ddagger$ School of Electrical Engineering and Computer Science, University of Ottawa, Ontario, Canada\\
$^{\dagger}$The Interdisciplinary Centre for Security Reliability and Trust (SnT), University of Luxembourg, Luxembourg  \\
Emails: lkong2@uottawa.ca, \{steven.kisseleff, symeon.chatzinotas, bjorn.ottersten\}@uni.lu, melike.erolkantarci@uottawa.ca}}
\maketitle

\begin{abstract}
Reconfigurable Intelligent Surfaces (RIS) are planar structures connected to electronic circuitry, which can be employed to steer the electromagnetic signals in a controlled manner. Through this, the signal quality and the effective data rate can be substantially improved. While the benefits of RIS-assisted wireless communications have been investigated for various scenarios, some aspects of the network design, such as coverage, optimal placement of RIS, etc., often require complex optimization and numerical simulations, since the achievable effective rate is difficult to predict. This problem becomes even more difficult in the presence of phase estimation errors or location uncertainty, which can lead to substantial performance degradation if neglected. Considering randomly distributed receivers within a ring-shaped RIS-assisted wireless network, this paper mainly investigates the effective rate by taking into account the above-mentioned impairments. Furthermore, exact closed-form expressions for the effective rate are derived in terms of Meijer's $G$-function, which (i) reveals that the location and phase estimation uncertainty should be well considered in the deployment of RIS in wireless networks; and (ii) facilitates future network design and performance prediction. 
\end{abstract}

\IEEEpeerreviewmaketitle

\section{Introduction}
\IEEEPARstart{R}{econfigurable} intelligent surfaces (RISs) are considered key enablers for  future wireless communications \cite{liaskos2018new}. Accordingly, RISs have been proposed for various scenarios and applications including 6G \cite{8981888}, internet of things (IoT) \cite{mursia2020risma}, smart cities \cite{9253607}, challenging environments \cite{9600850}, etc. The benefit of RIS lies in its capability of shaping the wireless propagation environments by adjusting the signal reflections \cite{di2019smart}. Through this, the signal quality and connectivity can be substantially improved. Furthermore, the energy consumption at RIS is extremely low, which is a favorable property compared to traditional relaying.

Various aspects of RIS-assisted networks have been investigated so far, such as optimal placement of RIS \cite{9386246}, signal routing using multiple RIS \cite{9241752}, hybrid automatic repeat request for RIS-aided communication systems \cite{aihybrid}, etc. Due to the associated system complexity, the methods proposed in these works typically employ highly complex and computationally intensive optimization techniques. For the same reason, it is difficult to predict the system performance for RIS-assisted networks without accurate and thorough numerical evaluation. On the other hand, if an important performance metric, such as effective rate, can be directly determined without prior simulations of the environment merely based on the known network configuration, the network design and planning can be substantially simplified. This motivates the analytical derivation of such metrics, which is the focus of this work.

Effective rate connecting both physical layer and link layer, is a popular performance metric used in the RIS-aided wireless networks. \cite{9312154} describes the RIS-assisted links using the cascaded fading channel model and also provides a derivation of the effective rate with exact and approximated expressions. The authors in \cite{AMAN2021101339} study the effective rate under two extreme assumptions with regard to the availability of the channel state information. However, the impact of phase estimation errors and location uncertainty on the effective rate has not been addressed yet. This paper aims at bridging this gap by providing analytical expressions for the effective rate with consideration of the aforementioned impairments.

The contributions of this work are three-fold:
\begin{itemize}
\item analytical closed-form expressions for the effective rate are derived in terms of Meijer's $G$-function;
\item the obtained solution incorporates important impairments, which are likely to occur in practical scenarios of RIS-assisted networks, i.e. location uncertainty and imperfect phase estimation;
\item extensive simulations validate the accuracy of the derived closed-form expressions independently from the network configuration, impairments, and number of RIS elements.
\end{itemize}

The remainder of this paper is organized as follows. In Section II, the system and signal propagation models are described. In Section III, the closed-form expressions for the effective rate under impairments are derived. Section IV provides insights upon the numerical evaluation of the derived expressions and the predicted system performance under various conditions. Subsequently, Section V concludes the paper.

\textit{Mathematical Functions and Notations}: $j = \sqrt{-1}$, $\mathbf{K}_n(\cdot)$ is the modified Bessel function of the second kind \cite[Eq. (5.52)]{gradshteyn2014table}. $\Gamma(L)$ is the gamma function \cite[Eq. (8.310.1)]{gradshteyn2014table}, $G_{p,q}^{m,n}[\cdot]$ is the univariate Meijer's $G$-function\footnote{It is worth mentioning that the univariate Meijer's $G$-function is a built-in function and available at MATLAB, i.e., $\text{meijerG}(a,b,c,d,x)$, Mathematica, and Maple. It can be flexibly transformed into many functions, e.g. the exponential, logarithm, incomplete gamma functions, etc. More specifically, it has been proved useful and generic when analyzing outage probability, bit error rate, channel capacity, secrecy outage probability, etc. Interested readers can find examples in \cite{7089281,8405529}.} \cite[Eq. (9.301)]{gradshteyn2014table}. $H_{p,q}^{m,n}[\cdot]$ is the Fox's $H$-function \cite[Eq. (8.3.1.1)]{prudnikov1990integrals}. $\mathbb{E}[x]$ is the expectation operator over a random variable (RV) $x$.
\section{System model and preliminaries}
\begin{figure} 
\centering{\includegraphics[width=\columnwidth]{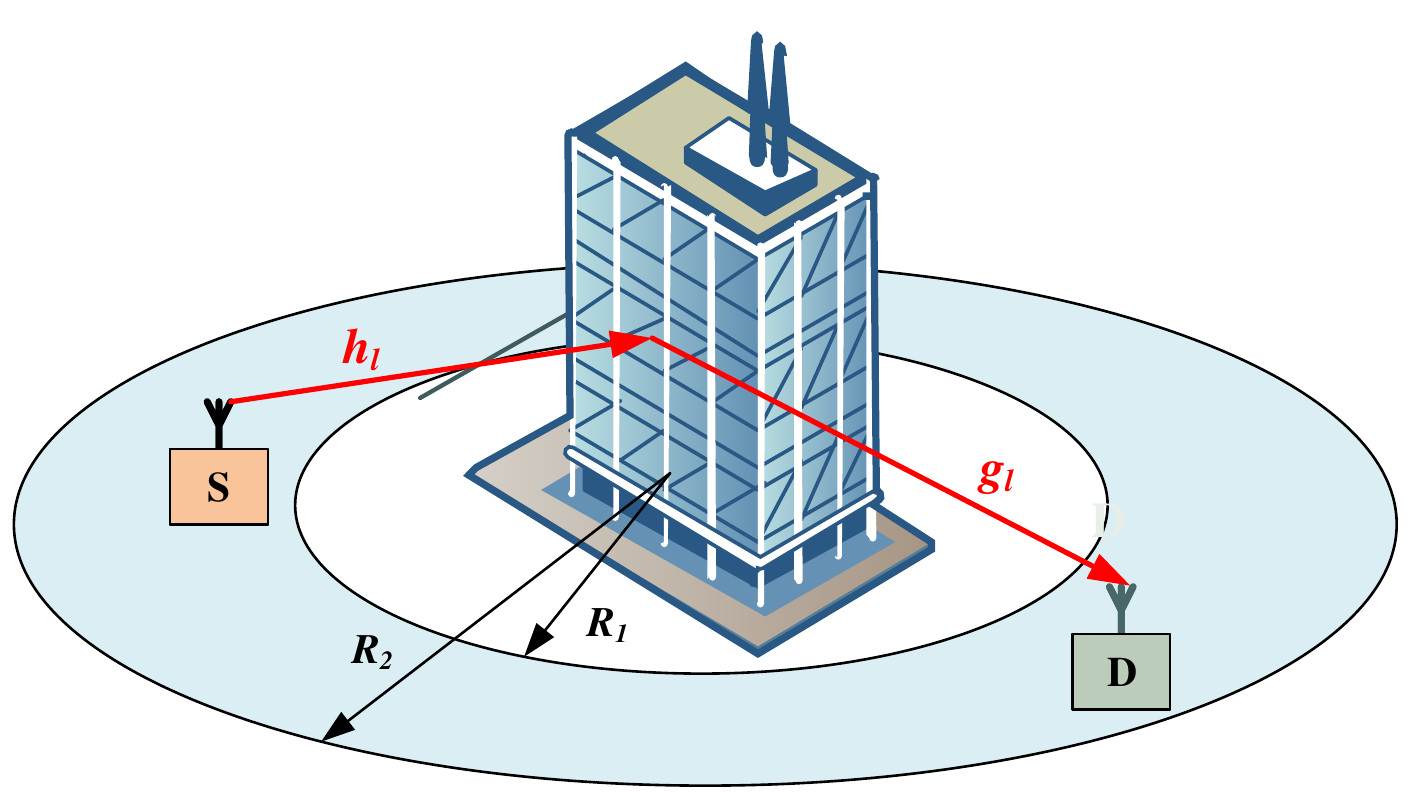}}
\caption{The system model.}
\label{Fig_systemModel}
\end{figure}
We consider a single-input single-output (SISO) communication system assisted by a single RIS, embedded on the surface of a building, shown in Fig. \ref{Fig_systemModel}. Specifically, the source node (\textbf{S}) transmits signal $s$, which is reflected by RIS toward the destination node (\textbf{D}).
The RIS comprises $L$ passive and low-cost reflecting elements. We assume that the direct links are much weaker than the reflection-based ones due to obstacles and fading effects \cite{9138463}. Thus, the direct links are omitted in our investigation. In addition, the exact location of \textbf{D} and its distance to RIS might be unknown. As such, \textbf{D} is assumed to be randomly located within a ring centered at the RIS's position according to homogeneous Poisson point processes (HPPP), where the ring outer radius $R_2$ indicates its coverage region, and the inner radius $R_1$ depicts the minimum distance.   

The instantaneous received signal at \textbf{D} is given by
\begin{equation} \label{Signal_D}
y_{D} = \sqrt{P} \left \vert \sum \limits_{l=1}^L  h_l u_l g_l \right \vert s + \mathfrak{n},
\end{equation}
where $P$ is the transmit power, $h_l=\mathfrak{h}_l \exp(-j\phi_l) $ and $g_l =\mathfrak{g}_l \exp(-j\psi_l) $ are independent and identically distributed (i.i.d.) Rayleigh RVs, which correspond to the channel coefficients from \textbf{S} to the $l$-th reflector element and the $l$-th reflector element to \textbf{D}, respectively. $\mathfrak{h}_l$, $\mathfrak{g}_l$, $\phi_l$, and $\psi_l$ denote the amplitudes and phases of the corresponding fading channel gains. $u_l = w_l(\theta_l)\exp(j\theta_l)$ is the reflection coefficient produced by the $l$-th element of the RIS, herein $w_l(\theta_l)=1$ for the ideal phase shifts, $l=1,\cdots,L$. $s$ is the transmitted signal with unit energy. $\mathfrak{n}$ is the additive white Gaussian noise (AWGN) with zero mean and $\sigma^2$ variance.

Based on (\ref{Signal_D}), the signal-to-noise ratio (SNR) observed at $\textbf{D}$ is obtained as
\begin{align}
\gamma_D& = \frac{P}{d_{SR}^\delta d_{RD}^\delta \sigma^2 } \left \vert \sum \limits_{l=1}^L  \mathfrak{h}_l \mathfrak{g}_l \exp\big( j(\theta_l - \phi_l - \psi_l)\big) \right \vert ^2 \nonumber \\
&=\frac{\rho}{ d_{SR}^\delta d_{RD}^\delta} \left \vert \sum \limits_{l=1}^L \mathfrak{h}_l \mathfrak{g}_l \exp(j\epsilon_l) \right\vert^2,
\end{align}
where $\rho = \frac{P}{\sigma^2}$ and $\delta$ is the path loss exponent. $d_{SR}$ and $d_{RD}$ denote the distances from \textbf{S} to RIS and from RIS to \textbf{D}, respectively. The phase estimation error $\epsilon_l$ is assumed to be an i.i.d. RV following the uniform distribution. Here, \textbf{S} may employ a cheap off-the-shelf hardware, which has a very poor phase estimation accuracy, i.e., $\epsilon_l \sim U(-\pi,\pi)$. For simplicity of notations, we fix $d_{SR}$ in our system setup, and let $\bar{\gamma} = \frac{\rho}{d_{SR}^\delta }$. Next, we would like to determine the probability density of $\textbf{D}$, which is needed for the calculation of the effective rate.
\begin{proposition}
The probability density function (PDF) of $\gamma_D$ is given by (\ref{PDF_MISO}), shown at the top of next page.
\begin{figure*}[!t]
\setcounter{MYtempeqncnt}{\value{equation}}
\setcounter{equation}{2}
\begin{align} \label{PDF_MISO}
f_D(\gamma) = \frac{1}{2\delta \bar{\gamma}\Gamma(L)(R_2^2-R_1^2)}\left[R_2^{2+\delta}G_{1,3}^{2,1} \left[ {\frac{R_2^\delta\gamma}{4\bar{\gamma}}   \left| {\begin{array}{*{20}c}
		{-\frac{2}{\delta}}   \\
		{L-1,0,-1-\frac{2}{\delta} }  \\
		\end{array}} \right.} \right] - R_1^{2+\delta}G_{1,3}^{2,1} \left[ {\frac{R_1^\delta\gamma}{4\bar{\gamma}}   \left| {\begin{array}{*{20}c}
		{-\frac{2}{\delta}}   \\
		{L-1,0,-1-\frac{2}{\delta} }  \\
		\end{array}} \right.} \right]\right],
\end{align}
\hrulefill
\end{figure*}
\end{proposition}
\begin{proof}
Let $X = \left \vert \sum \limits_{l=1}^L \mathfrak{h}_l \mathfrak{g}_l \exp(j\epsilon_l) \right\vert^2$, using the result in \cite[Eq. (43)]{9138463}, the cumulative distribution function (CDF) of $X$ is given by
\begin{align} \label{CDF_gammaB}
F_X(x) = 1- \frac{2^{1-L}}{\Gamma(L) }\left( \sqrt{x} \right)^L \mathbf{K}_L\left(  \sqrt{x } \right).
\end{align}
Next, by taking the derivative of (\ref{CDF_gammaB}) with respect to $x$, and then using $\frac{d}{dz}[z^L\mathbf{K}_L(z)] = -z^L \mathbf{K}_{L-1}(z) $ from \cite[Eq. (8.486.14)]{gradshteyn2014table}, we obtain the PDF of $X$ with  
\begin{align} \label{PDF_B}
f_X(x) &= \frac{1}{2^L  \Gamma(L)}\left( \sqrt{x} \right)^{L-1} \mathbf{K}_{L-1}\left(  \sqrt{x} \right) \nonumber\\
& \mathop=^{(a)} \frac{1}{4 \Gamma(L) }G_{0,2}^{2,0} \left[ { \frac{x }{4}\left|  {\begin{array}{*{20}c}
    {-}   \\
   { L-1,0}  \\
\end{array}} \right.}  \right],
\end{align}
where step $(a)$ is developed using the technique \cite[Eq. (9.34.3)]{gradshteyn2014table}.
With \cite[Eq. (28)]{9622149} in mind, the PDF of $d_{RD}$ is given by
\begin{align}
f_d(r) = \frac{2r}{R_2^2 - R_1^2}, \quad R_1 \le r \le R_2,
\end{align}
Correspondingly, the PDF of $Y = d_{RD}^\delta$ is 
\begin{align} \label{PDF_PL}
f_Y(y) = \frac{2y^{\frac{2}{\delta}-1}}{\delta(R_2^2 - R_1^2)}, \quad R_1^\delta \le r \le R_2^\delta.
\end{align}

Subsequently, plugging (\ref{PDF_B}) and (\ref{PDF_PL}) into $\gamma_D =  \bar{\gamma}\frac{X}{Y}$ leads to
\begin{align}
f_D(\gamma) = \frac{1}{\bar{\gamma}} \int_{R_1^\delta}^{R_2^\delta} y f_X\left( \frac{\gamma y}{\bar{\gamma}} \right)f_Y(y)dy.
\end{align}
Then, using \cite[Eq. (1.16.4)]{prudnikov1990integrals}, we derive (3) and the proof is achieved.
\end{proof}

\section{Effective rate}
The effective rate is a dual concept connecting both physical layer and link layer. It is widely used to determine the maximal effective bandwidth that the system can support for a given delay constraint \cite{1210731}. Assuming the block fading channel for our system configuration, the normalized effective rate at \textbf{D} is analytically given by
\begin{equation}
\mathcal{R} = -\frac{1}{A} \log_2\underbrace{\left[\mathbb{E}\left((1+\gamma_M)^{-A}\right)\right]}_\mathcal{M} \text{bit/s/Hz},
\end{equation}
where $A = \frac{\theta T B}{\ln 2}$, where $\theta$, $T$, and $B$ represent the asymptotic decay rate of the buffer occupancy, the block length, and the system bandwidth, respectively. 
\begin{proposition}
Considering the existence of phase estimation error and unknown exact location, the effective rate for the given system configuration is given by
\begin{align} \label{ER_exact}
\mathcal{R} = -\frac{1}{A} \log_2\left[ \frac{\mathcal{W}}{2\delta \bar{\gamma} \Gamma(L)\Gamma(A)(R_2^2-R_1^2)} \right],
\end{align}
where
\begin{align*}
\mathcal{W} =& R_2^{2+\delta}G_{2,4}^{3,2} \left[ {\frac{R_2^\delta}{4\bar{\gamma}}   \left| {\begin{array}{*{20}c}
		{-\frac{2}{\delta},0}   \\
		{L-1,0,A-1,-1- \frac{2}{\delta} }  \\
		\end{array}} \right.} \right]  \nonumber \\
		& - R_1^{2+\delta}G_{2,4}^{3,2} \left[ {\frac{R_1^\delta}{4 \bar{\gamma}}   \left| {\begin{array}{*{20}c}
		{-\frac{2}{\delta},0}   \\
		{L-1,0,A-1,-1-\frac{2}{\delta} }  \\
		\end{array}} \right.} \right].
\end{align*}
\end{proposition}
\begin{proof}
The proof starts by re-expressing $\frac{1}{(1+x)^c}$ in terms of Meijer's $G$-function \cite[Chpt. 8.4]{prudnikov1990integrals}
\begin{align} \label{logarithm}
\frac{1}{(1+x)^c} & = \frac{1}{\Gamma(c)}G_{1,1}^{1,1} \left[ {x  \left| {\begin{array}{*{20}c}
		{1-c}   \\
		{0}  \\
		\end{array}} \right.} \right] \nonumber \\
		& \mathop=^{(b)} \frac{1}{\Gamma(c)}H_{1,1}^{1,1} \left[ {x  \left| {\begin{array}{*{20}c}
		{(1-c,1)}   \\
		{(0,1)}  \\
		\end{array}} \right.} \right],
\end{align}
where step $(b)$ is developed using the property of Fox's $H$-function \cite[Eq. (8.3.2.21)]{prudnikov1990integrals}. Next, we insert (\ref{logarithm}) and (\ref{PDF_MISO}) into $\mathcal{M}$, the expression $\mathcal{M}$ can be easily stated as
\begin{align} \label{Proof_Proposition2}
    \mathcal{M} =&  \frac{1}{2\delta \bar{\gamma}\Gamma(L)\Gamma(A)(R_2^2-R_1^2)} \int_0^\infty G_{1,1}^{1,1} \left[ {\gamma  \left| {\begin{array}{*{20}c}
		{1-A}   \\
		{0}  \\
		\end{array}} \right.} \right]   \nonumber \\
		&\times \left[R_2^{2+\delta}G_{1,3}^{2,1} \left[ {\frac{R_2^\delta\gamma}{4\bar{\gamma}}   \left| {\begin{array}{*{20}c}
		{-\frac{2}{\delta}}   \\
		{L-1,0,-1-\frac{2}{\delta} }  \\
		\end{array}} \right.} \hspace{-1.2ex} \right] \right. \nonumber \\
		& \left.\hspace{4ex}- R_1^{2+\delta}G_{1,3}^{2,1} \left[ {\frac{R_1^\delta\gamma}{4\bar{\gamma}}   \left| {\begin{array}{*{20}c}
		{-\frac{2}{\delta}}   \\
		{L-1,0,-1-\frac{2}{\delta} }  \\
		\end{array}} \right.} \right]\right],
\end{align}
subsequently performing the Mellin transform of the product of two Fox's $H$-function \cite[Eq.(2.25.1)]{prudnikov1990integrals}, the proof is eventually concluded.  
\end{proof}
\begin{remark}
Considering the phase estimation error alone, the effective rate for the given system is given by
\begin{align} \label{ER_ErrorOnly}
    \mathcal{R}= -\frac{1}{A}\log_2\left[\frac{1}{4\bar{\gamma}\Gamma(A)\Gamma(L)}G_{1,3}^{3,1} \left[ {\frac{1}{4\bar{\gamma}} \left| {\begin{array}{*{20}c}
		{0}   \\
		{L-1.0,A-1}  \\
		\end{array}} \right.} \hspace{-1.5ex} \right] \right].
\end{align}
\end{remark}
\begin{proof}
Fixing both $d_{SR}$ and $d_{RD}$, the derivation of (\ref{ER_ErrorOnly}) is based on the PDF of $\gamma_D = \bar{\gamma} X$, which is obtained as follows
\begin{align} \label{PDF_WithError}
    f_D(\gamma) = \frac{1}{\bar{\gamma}}f_X\left(\frac{\gamma}{\bar{\gamma}} \right).
\end{align}
Next plugging (\ref{PDF_WithError}) and (\ref{logarithm}) into $\mathcal{M}$, and following the similar line as in (\ref{Proof_Proposition2}), the proof is finished.
\end{proof}
\begin{remark}
When $\rho \rightarrow \infty$, i.e., at high SNR regime, the asymptotic behavior of ER given in (\ref{ER_exact}) is further simplified as 
\begin{align}
\label{eq:asymptotic}
    \mathcal{R}^\infty \approx -\frac{1}{A} \log_2 \left( \frac{R_2^{2 + \delta} -R_1^{2+\delta}}{ \bar{\gamma} \left(2\delta + 4 \right) (L-1)(A-1)(R_2^2 -R_1^2)}\right).
\end{align}
\end{remark}
\begin{proof}
When $\rho$ is at high SNR regime, both $\frac{R_2^{\delta}}{4\bar{\gamma}}$ and $\frac{R_1^{\delta}}{4\bar{\gamma}}$ tend to $0$. 
Thus, the expression in \eqref{eq:asymptotic} is obtained by expanding the Meijer's $G$-function at $0$ and using the fact given in \cite{7370883}.
\end{proof}
\section{Numerical Analysis} 
In this section, we provide numerical evaluations of the theoretical ER and compare it with the results obtained using Monte Carlo simulations with $10^6$ runs.

We start with the analysis of effective rate for different values of $\rho$ and $L$. For this, we assume $A = 5.4$, $\delta =3.4$, $R_1=2$ m and $R_2 = 5$ m. The results are depicted in Fig. \ref{Fig2}. We observe that the analytically determined effective rate is close to the results obtained via extensive simulations.  
\begin{figure}[!t]
\centering{\includegraphics[width=\columnwidth]{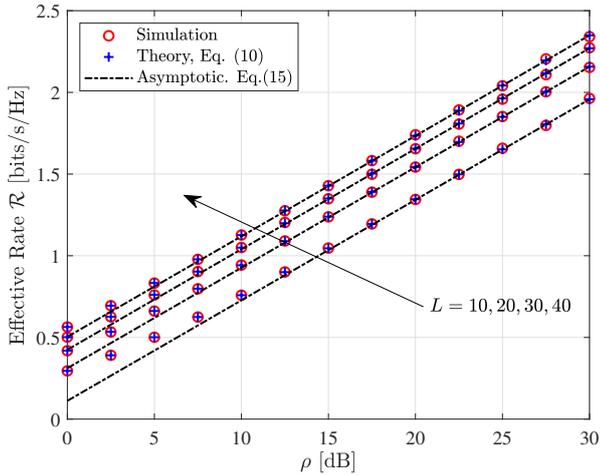}}
\caption{Effective rate $\mathcal{R}$ vs $\rho$ for selected $L$.}
\label{Fig2}
\end{figure}
Furthermore, we observe that the effective rate increases with increasing number of RIS elements. This is expected, since the spatial diversity increases with increasing $L$, which can be used to improve the performance. In addition, we observe that the results of the asymptotic approximation deviate substantially from the simulation and theory at low SNR regime. Specifically, the approximation underestimates the effective rate, especially with a low number of RIS elements. However, with practical value of $L$, this effect vanishes, such that the low complexity expression in \eqref{eq:asymptotic} becomes accurate even at low SNR.

\begin{figure}[!t] 
\hfill
\subfloat[$d_{RD} = 2$m]{\includegraphics[width=\columnwidth]{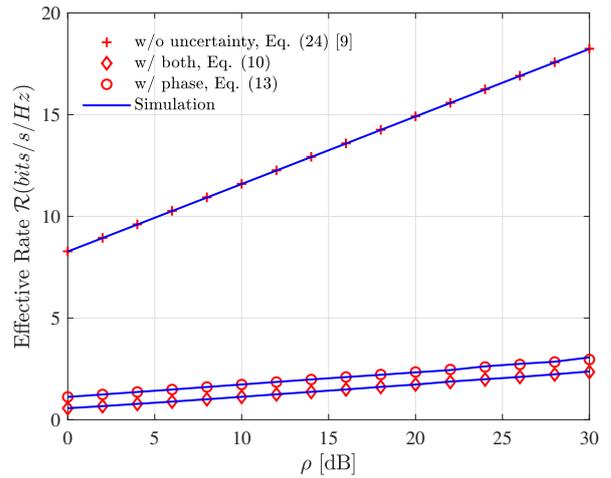}}
\hfill
\subfloat[$d_{RD} = 5$m]{\includegraphics[width=\columnwidth]{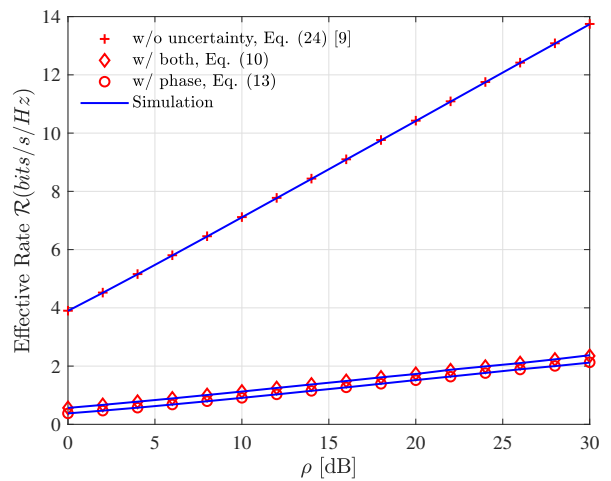}}
\hfill
\caption{Effective rate $\mathcal{R}$ vs. $\rho$ (a) $d_{RD} =2$; and (b) $d_{RD} =5$.}
\label{Fig6}
\end{figure}

In Fig. \ref{Fig6}, the results are depicted for three scenarios assuming $L=40$: (i) without impacts from location and phase error uncertainty; (ii) with phase error only; and (iii) with both uncertainties. We observe a substantial performance degradation in terms of effective rate in presence of the phase and location uncertainties. Also, the location uncertainty affects the effective rate performance much less than the phase uncertainty. Interestingly, the result obtained in \cite{9312154} without taking into account these practical impairments leads to a dramatic overestimation of the expected system performance, which may even render the whole network design based on such an idealistic effective rate prediction incorrect. 

Fig. \ref{Fig3} investigates the impact of the coverage area on the effective rate. Here, $\rho=5$ dB has been assumed. One can observe that the effective rate decreases rapidly with increasing coverage area, since the average distance between the user and RIS increases, thus increasing the path loss. In practice, Fig. \ref{Fig3} being a theoretical proof, can be utilized to deduce e.g. the required number of RIS elements in order to satisfy a given network demand within a certain coverage area. Consequently, no complicated numerical optimization is required for this design problem anymore.

\begin{figure}[!t] 
\centering{\includegraphics[width=\columnwidth]{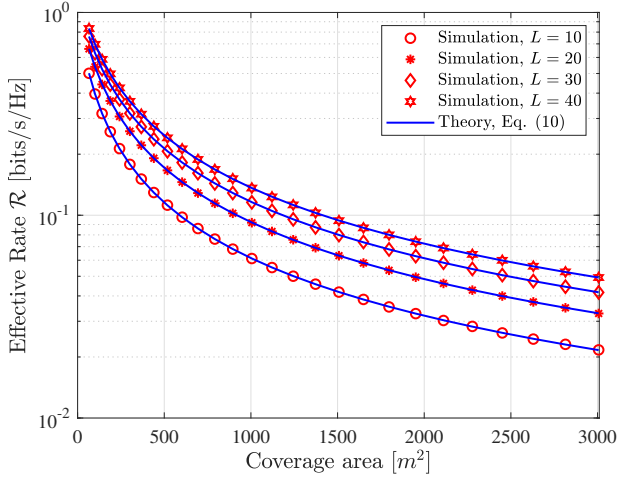}}
\caption{Effective rate $\mathcal{R}$ vs. coverage area $\pi(R_2^2 - R_1^2)$ for selected $L$.}
\label{Fig3}
\end{figure}

Considering four different kinds of mobile radio environments \cite[Chapter 3.9]{rappaport1996wireless}, Fig. \ref{Fig4} plots $\mathcal{R}$ against the decay rate $\theta$ for selected values of path loss exponent $\delta$. For this, we fix $L=20$, $R_2=8$ m, and $A = \frac{\theta TB}{\ln 2}$ with $TB=1$ while all other parameters remain unchanged. Obviously, the effective rate demonstrates a performance degradation as $\theta$ increases. In other words, a smaller delay constraint results in a higher effective rate performance. Correspondingly, a better quality of service can be supported for delay-tolerant services. Starting from $\theta\approx 10^{0}$ 1/bit, the effective rate starts to drop to very low values, i.e., the delay-sensitive communications services, such as video conferencing or gaming, cannot be supported. In addition, we observe that for a large value of $\delta$, i.e., in case of poor propagation conditions, the performance drops significantly compared to low path loss exponents.

\begin{figure} 
\centering{\includegraphics[width=\columnwidth]{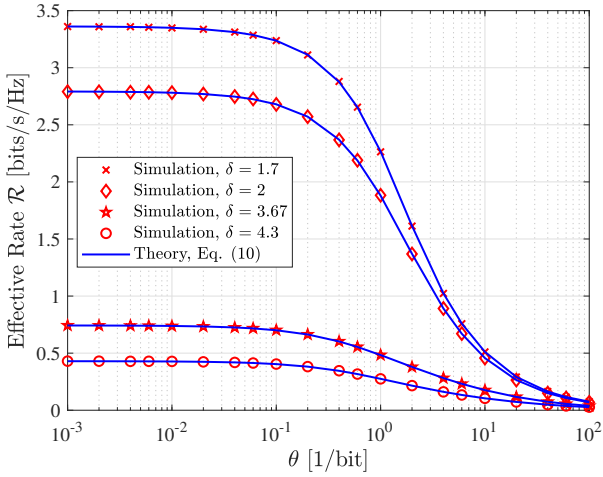}}
\caption{Effective rate $\mathcal{R}$ vs. $\theta$ for selected $\delta$.}
\label{Fig4}
\end{figure}


\section{Conclusion}
In this paper, we derived two novel closed-form expressions of effective rate for RIS-assisted wireless networks under the impacts of location uncertainty and phase errors. Our derived closed-form results are validated via Monte-Carlo simulations and their accuracy is displayed to be high. We further observed a performance degradation in case of increasing decay rate and compared the system performance in different environments. Using the derived expression, it is possible to facilitate the network designers with respect to the coverage area, placement of RIS, etc. under realistic conditions of phase errors and location uncertainty.


\section*{Acknowledgment}
This work has been supported by the Luxembourg National Research Fund (FNR) projects, entitled Exploiting Interference for Physical Layer Security in 5G Networks (CIPHY), and Reconfigurable Intelligent Surfaces for Smart Cities (RISOTTI). This work was done while Long Kong was a research associate at the University of Luxembourg. He is now a postdoctoral fellow at the University of Ottawa.
\ifCLASSOPTIONcaptionsoff
  \newpage
\fi

\balance

\bibliographystyle{IEEEtran}
\bibliography{EC_RIS_2021}

\end{document}